\documentclass[preprint,pra,superscriptaddress,floatfix,letterpaper]{revtex4}
\usepackage{graphicx}
\usepackage{subfigure}
\usepackage{amssymb}
\usepackage{verbatim}
\usepackage{amsmath}
\usepackage{amscd}
\usepackage{amssymb}
\usepackage{setspace}

\newtheorem{thm}{Theorem}

\newenvironment{proof}{{\noindent{\bf Proof:}}}{$\hfill\Box$}

\def\NP{{\sf{NP}}}

\def\Tr{{\text{tr}}}
\def\T{{\text{T}}}
\def\SEP{{\text{SEP}}}
\def\PPT{{\text{PPT}}}

\begin{document}

\singlespacing

\title{Approximating the Set of Separable States Using the Positive Partial Transpose Test}

\author{Salman Beigi}
\affiliation{Institute for Quantum Information, California Institute of
Technology, Pasadena, CA 91125, USA}

\author{Peter W. Shor}
\affiliation{Department of Mathematics, Massachusetts Institute of
Technology, Cambridge, MA 02139, USA}

\begin{abstract}
The positive partial transpose test is one of the main criteria
for detecting entanglement, and the set of states with positive
partial transpose is considered as an approximation of the set of
separable states. However, we do not know to what extent this
criterion, as well as the approximation, are efficient. In this
paper, we show that the positive partial transpose test gives no
bound on the distance of a density matrix from separable states.
More precisely, we prove that, as the dimension of the space
tends to infinity, the maximum trace distance of a positive
partial transpose state from separable states tends to $1$. Using
similar techniques, we show that the same result holds for other
well-known separability criteria such as reduction criterion,
majorization criterion and symmetric extension criterion. We also
bring an evidence that the sets of positive partial transpose states
and separable states have totally different shapes.
\end{abstract}

\date{\today}

\maketitle




\section{Introduction}

The problem of detecting entanglement has been focused in quantum
information theory for many years. The problem is: given a
bipartite mixed state $\rho_{AB}$, decide whether this state is
entangled or separable. The first attack toward solving this
problem is the following observation due to Peres and the
Horodeckis \cite{peres, 23}. If $\rho_{AB}=\sum_i
p_i\,\rho_{A}^{(i)}\otimes \rho_{B}^{(i)} $ is separable, then
$(\rho_{AB})^{\T_{\!B}}=\sum_i p_i\rho_{A}^{(i)}\otimes (\rho_{B}^{(i)})^\T$,
where $M^{\T}$ denotes the transpose of matrix $M$, is also a
quantum state, and is a positive semi-definite matrix.
Therefore, if $\rho_{AB}$ is separable, its partial transpose,
$(\rho_{AB})^{\T_{\!B}}$, should be positive semi-definite. The
Horodeckis have proved that this criterion characterizes all
separable states in dimensions $2\times 2$ and $2\times 3$
\cite{23}. However, there are entangled states in dimension
$3\times 3$ with a positive partial transpose \cite{bound}.

Although the set of positive partial transpose states (PPT states)
does not coincide with the set of separable states, it is usually
considered as an approximation of this set. For example in
\cite{distance} instead of estimating the distance from separable states, the distance of an arbitrary state from PPT
states has been computed as an ``strongly related problem." Also in \cite{geometry} the geometry of the set of PPT
states has been studied to understand the properties of the set of
separable states. However, we do not know how efficient these
approximations are. For instance, given an upper bound on the
distance of a state from PPT states, does it give an upper bound
on the distance from separable states?

We can think of this problem from the point of view of complexity
theory. Gurvits \cite{gurvits} has proved that given a bipartite
density matrix $\rho_{AB}$, it is $\NP$-hard to decide whether
this state is separable or entangled. An approximate formulation
of this problem is the following: given a bipartite density matrix
$\rho_{AB}$ and $\epsilon>0$, decide whether there exists a
separable state in the $\epsilon$-neighborhood (in trace
distance) of $\rho_{AB}$. Gurvits has established a reduction
from Knapsack to this problem, and has proved the $\NP$-hardness
of the separability problem, but only for exponentially small
$\epsilon$. However, as mentioned in \cite{salman} by replacing
Knapsack with 2-out-of-4-SAT and repeating a similar argument, the $\NP$-hardness can be proved for an
inverse polynomial $\epsilon$. Also, Gharibian \cite{sevag} has
shown the same result using a reduction from the Clique problem.
Now the question is that how large $\epsilon$ can be while
getting to the $\NP$-hardness. For example, is there an efficient
algorithm to decide whether the distance of a given state from
separable states is less than $1/3$, or it is an $\NP$-hard
problem? Equivalently, is there an efficiently implementable separability test such that if
a state passes the test, then it is $1/3$-close to the set of
separable states?

In this paper we consider the converse of this question, i.e.
given a separability criterion, if a state passes this test, can we
claim a non-trivial upper bound on the distance of this state
from separable states? We prove that the answer
for the PPT criterion, as well as other well-known separability
tests such as reduction criterion \cite{reduction}, majorization
criterion \cite{majorization}, and symmetric extension criterion
\cite{symmetric1, symmetric2}, is no. More precisely, we prove the
following theorem.

\begin{thm}\label{thm:main} Let $\mathcal{H}$ be a bipartite
Hilbert space. For every $\varepsilon>0$, if the dimension of each subsystem of $\mathcal{H}$
is large enough, there exists a PPT state acting on $\mathcal{H}$
whose trace distance from separable states is at least
$1-\varepsilon$.

\end{thm}

To the best knowledge of authors, this is the first result that compares separable
states relative to PPT states in terms of their distance. However, the volume of these sets has been
studied by several authors. For instance, by estimating the volume of separable states and PPT states in the Hilbert-Schmidt norm,
it has been shown in \cite{aubrun} that a random PPT state is entangled. The same conclusion has been proved in \cite{ye} in terms of Bures volume. See also \cite{n-partite} and \cite{swz08} for some other results in this setting.

\subsection{Main ideas}

Let $\mathcal{H}=\mathcal{H}^A \otimes \mathcal{H}^B$ be a
bipartite Hilbert space. We want to find PPT states $\rho^{(n)}
\in \mathcal{H}^{\otimes n}$ such that the trace distance of
$\rho^{(n)}$ from separable states is close to 1, for enough large
numbers $n$. Suppose $\rho$ is an entangled PPT state. Then
$\rho^{\otimes n}$ is entangled and also PPT. We claim that the
sequence of states $\rho^{(n)}=\rho^{\otimes n}$ works for us.
The intuition is that for two different quantum states $\rho$ and
$\sigma$, the trace distance of $\rho^{\otimes n}$ and
$\sigma^{\otimes n}$ tends to 1 as $n$ tends to infinity.
However, in this problem $\sigma$ is not a fixed state and ranges
over all separable states. Also, it is not obvious (and may not
hold) that the closest separable state to
$\rho^{\otimes n}$ is of the form $\sigma^{\otimes n}$. (If we replace the trace distance with
$E_R(\rho)$, the relative entropy of entanglement, this property
does not hold \cite{relent}.)

Another idea is to use entanglement distillation. Suppose the
state $\rho$ is distillable. It means that, having arbitrary many
copies of $\rho$, using local quantum operations and classical communications (LOCC maps), we can obtain as many
EPR pairs as we want (say $m$ pairs). LOCC maps send separable states to
separable states, and the trace distance decreases under trace
preserving quantum operations. Therefore, the distance of
$\rho^{\otimes n}$ from separable states is bounded from below by
the distance of EPR$^{\otimes m}$ from separable states, which we
know is close to 1 for large numbers $m$. Therefore, if $\rho$ is
distillable, the trace distance of $\rho^{\otimes n}$ from
separable states tends to 1 as $n$ tends to infinity.

It is well-known that PPT states are not distillable under LOCC
maps. So we cannot use this idea directly. On the other hand, in
this argument, the only property of LOCC maps that we use, is
that they send separable states to separable states. Thus we may
replace LOCC maps with {\it non-entangling maps}, the maps that
send every separable state to a separable state. Due to the
seminal work of Brandao and Plenio \cite{brandao1, brandao2} every
entangled state is distillable under {\it asymptotically}
non-entangling maps. As a result, by replacing LOCC maps with
asymptotically non-entangling maps and repeating the previous
argument, we conclude that the trace distance of $\rho^{\otimes
n}$ from separable states tends to 1.

Although this idea gives a full proof of Theorem \ref{thm:main},
we do not present it in this paper. Instead, we use more
fundamental techniques, namely, {\it quantum state tomography} and
{\it quantum de Finetti theorem} \cite{renner, povm}. In fact,
these two techniques are the basic ideas of the results of
\cite{brandao1, brandao2} that we mentioned above. Since
$\rho^{\otimes (n+k)}$ is a symmetric state, we may assume that
the closest separable state to $\rho^{\otimes (n+k)}$ is also
symmetric. Then by tracing out $k$ registers and using the
finite quantum de Finetti theorem we conclude that the trace
distance of $\rho^{\otimes (n+k)}$ from separable states is lower
bounded by the trace distance of $\rho^{\otimes n}$ from
separable states of the form
\begin{equation} \label{eq:1}
\sum_i p_i\, \sigma_i^{\otimes n}.
\end{equation}
Since such a state is separable and $\rho$ is entangled, the
sum of $p_i$'s for which $\sigma_i$ is close to $\rho$ cannot be
large. On the other, if $\sigma_i$ is far from $\rho$, using quantum state
tomography one can distinguish $\rho^{\otimes n}$ from
$\sigma_i^{\otimes n}$.
Putting these two points together we show that the trace distance of $\rho^{\otimes n}$ and a
separable state of the form of (\ref{eq:1}) is close to $1$
for large enough $n$.

Note that in both of these arguments the only property of PPT
states that we use, is that if $\rho$ and $\sigma$ are PPT, then
$\rho \otimes \sigma$ is also PPT. So we can conclude the same
result for any separability test which satisfies this property.


\section{Preliminaries}\label{sec:prel}

A pure state $\vert \psi\rangle \in \mathcal{H}^A \otimes
\mathcal{H}^B$ is called separable if it can be written of
the form $\vert \psi\rangle = \vert \psi_A\rangle \otimes \vert
\psi_B \rangle$, where $\vert \psi_A\rangle \in \mathcal{H}^A$ and
$\vert \psi_B\rangle \in \mathcal{H}^B$. A density matrix acting
on $\mathcal{H}^A \otimes \mathcal{H}^B$ is called separable if
it can be written as a convex combination of separable pure
states $\vert \psi\rangle\langle \psi\vert$. We denote the set of
separable states by $\SEP$.

For two quantum states $\rho$ and $\sigma$ we denote their trace
distance by
\begin{equation}\label{eq:tr-dist}
\|\rho - \sigma\|_{\Tr}=\frac{1}{2}\, \Tr | \rho - \sigma |,
\end{equation}
where $|X|= \sqrt{X^{\dagger}X}$.

\subsection{Separability tests}

Assume that $\dim \mathcal{H}_A=\dim \mathcal{H}_B=d$, and fix an
orthonormal basis $\vert 1\rangle, \dots ,\vert d\rangle$ for both
of Hilbert spaces. The partial transpose of matrices acting
on $\mathcal{H}^A \otimes \mathcal{H}^B$ is a linear map defined
by $(M_A\otimes N_B )^{\T_{\!B}}=M_A\otimes N_B^\T$, where the transpose is taken with
respect to the fixed basis. Clearly, if
$\rho_{AB}$ is a separable state, $\rho_{AB}^{\T_{\!B}}$ is also a
density matrix and then positive semi-definite. However, it does
not hold for an arbitrary state. For example, the partial
transpose of the maximally entangled state is not positive
semi-definite; let $\Phi(d)$ to be the maximally
entangled state on $\mathcal{H}$
\begin{equation}\label{eq:epr} \Phi(d)=\frac{1}{d}\sum_{i,j=1}^d \vert i,i\rangle \langle j,
j\vert.
\end{equation}
$\Phi(d)$ is not positive semi-definite because
\begin{eqnarray*}
\Phi(d)^{\T_{\!B}} & = & \frac{1}{d} \sum_{i,j} \vert i\rangle \langle
j\vert \otimes \vert j\rangle\langle i\vert \\
& = & \frac{1}{d}I -\frac{1}{d} \sum_{i\neq j} \vert i\rangle
\langle i\vert \otimes \vert j\rangle \langle j\vert
+\frac{1}{d}\sum_{i\neq j} \vert i\rangle \langle j\vert \otimes
\vert j\rangle \langle i\vert\\
& = & \frac{1}{d}I - \frac{2}{d} \sum_{i<j} \vert
\phi_{ij}\rangle \langle \phi_{ij}\vert,
\end{eqnarray*}
where
\begin{equation} \label{eq:phi}
\vert \phi_{ij}\rangle=\frac{1}{\sqrt{2}}(\vert i\rangle\vert
j\rangle -\vert j\rangle \vert i\rangle).
\end{equation}
As a result, positive partial transpose is a test to detect
entanglement \cite{peres, 23}. More formally, if we denote the
set of density matrices with a positive semi-definite partial
transpose by PPT, then $\SEP\subseteq\PPT$.

Here is a list of some other separability criteria (see
\cite{survey, book}).

\begin{itemize}

\item Reduction criterion \cite{reduction}: $I\otimes \rho_B \geq
\rho_{AB}$, where $\rho_{B}=\Tr_{A} (\rho_{AB})$. Here, by $M\geq
N$ we mean $M-N$ is a positive semi-definite matrix.

\item Entropic criterion \cite{entropic}: $S_\alpha(\rho_{AB}) \geq S_\alpha(\rho_A)$
for $\alpha=2$ and in the limit $\alpha \rightarrow 1$, where
$S_\alpha(\rho)=\frac{1}{1-\alpha}\log \Tr(\rho^\alpha)$.

\item Majorization criterion \cite{majorization}: $\lambda_{\rho_{A}}^\downarrow \succ   \lambda_{\rho_{AB}}^\downarrow
$, where $\lambda_\rho^\downarrow$ is the list of eigenvalues of
$\rho$ in non-increasing order, and $y \succ x$ means that, for
any $k$, the sum of the first $k$ entries of list $x$ is less
than or equal to that of list $y$.

\item Cross norm criterion \cite{cross1, cross2}: $\Tr |
\mathcal{U}(\rho_{AB})|\, \leq 1$, where $\mathcal{U}$ is a linear
map defined by $\mathcal{U}(M\otimes N)=v(M)v(N)^{\T}$ and $v(X)=(col_1(X)^{\T}, \dots , col_{d}(X)^{\T} )^{\T}$,
where $col_i(X)$ is the $i$-th column of $X$.

\end{itemize}

All of these tests for separability are necessary conditions
but not sufficient. Doherty et al. \cite{symmetric1, symmetric2}
have introduced a hierarchy of separability criteria which are
both necessary and sufficient. Let $\rho_{AB}=\sum_i p_i\,
\sigma_i\otimes \tau_i$ be a separable state. Then
$$\rho_{AB_1B_2\cdots B_k} = \sum_i p_i\,
\sigma_i\otimes \tau_i^{\otimes k}$$
is an extension of
$\rho_{AB}$, meaning that $\rho_{AB}=\Tr_{B_2\cdots B_k}
(\rho_{AB_1\cdots B_k})$. Also it is symmetric, meaning that it does not
change under any permutation of subsystems $B_i$. More
precisely, for any permutation $\pi$ of $k$ objects, if we define
the linear map $P_\pi$ by $P_\pi \vert \psi_1\rangle \otimes
\cdots \otimes \vert \psi_k\rangle = \vert \psi_{\pi(1)}\rangle
\otimes \cdots \otimes \vert \psi_{\pi(k)}\rangle$, we have
\begin{equation}\label{eq:symmetric} P_{\pi}^{B_1\dots B_k}\, \rho_{AB_1B_2\cdots B_k}\,
(P_{\pi}^{B_1\dots B_k})^{\dagger}= \rho_{AB_1B_2\cdots B_k}.
\end{equation}
If such an extension exists, we say that $\rho_{AB}$ has a
symmetric extension to $k$ copies. Doherty el al. have proved that
a quantum state is separable if and only if it has a symmetric extension to
$k$ copies for any number $k$ \cite{symmetric1, symmetric2}.
Also, they have shown that the problem of checking whether a given
state has a symmetric extension to $k$ copies, for a fixed $k$,
can be expressed as a semi-definite programming, and can be solved
efficiently (however, the size of this semi-definite program grows exponentially
in terms of $k$). So we get to another separability test.

\begin{itemize}
\item Symmetric extension criterion \cite{symmetric1, symmetric2}: if $\rho_{AB}$ is separable, then
it has a symmetric extension to $k$ copies.
\end{itemize}

\subsection{Quantum state tomography}\label{sec:tomography}

An informationally complete POVM on $\mathcal{H}$ is a set of
positive semi-definite operators $\{M_n\}$ forming a basis for the
space of hermitian matrices on $\mathcal{H}$, and such that
$\sum_n M_n=I$. In \cite{povm} there is an explicit construction
of an informationally complete POVM in any dimension. Such a POVM
is useful for quantum state tomography.

Suppose $\{M_n^*\}$ is the dual of basis $\{M_n\}$, i.e.
$\Tr(M_nM_m^*)=\delta_{mn}$, where $\delta_{mn}$ is the Kronecker
delta function. For any hermitian operator $X$ we have
$$X=\sum_n \Tr(XM_n)\,M_n^*.$$
Therefore, having some copies of the state $\rho$, by measuring
$\rho$ using the POVM $\{M_n\}$, we can approximate $\Tr(\rho
M_n)$ and then find the matrix representation of $\rho$.

Assume that $\mathcal{H}=\mathcal{H}^A\otimes \mathcal{H}^B$ is a
bipartite Hilbert space. If $\{M_n^A\}$ and $\{M_m^B\}$ are
informationally complete POVM's on $\mathcal{H}^A$ and
$\mathcal{H}^B$ respectively, then
$\{M_n^A\otimes M_m^B\}$ is an informationally complete POVM on
$\mathcal{H}$. This means that, if the state $\rho_{AB}$ is shared
between two far apart parties, they can perform
quantum state tomography using classical communication. As a result, if the state $\rho_{AB}$ is
separable, then all intermediate states during the process are separable
as well.

\subsection{Quantum de Finetti theorem}

As in (\ref{eq:symmetric}), a quantum state $\rho^{(n)}$
acting on $\mathcal{H}^{\otimes n}$ is called symmetric if $P_\pi
\rho^{(n)} P_\pi^{\dagger}=\rho^{(n)}$ for any permutation $\pi$ of $n$
objects. A symmetric state is called {\it $k$-exchangeable} if it
has a symmetric extension to $n+k$ registers, i.e. a symmetric
state $\rho^{(n+k)}$ such that
$\Tr_{1,\dots,k}\,\rho^{(n+k)}=\rho^{(n)}$. Clearly, any state of
the form $\rho^{\otimes n}$ is $k$-exchangeable, for any $k$. Also
any convex combination of these states is $k$-exchangeable. {\it
Quantum de Finetti theorem} says that the converse of this
observation holds: if a state is $k$-exchangeable for
every $k$, it is in the convex hall of symmetric product states.

Quantum de Finetti theorem gives a characterization of
infinitely-exchangeable states. The following theorem, known as
the finite quantum de Finetti theorem, says that if a state is
$k$-exchangeable (but not necessarily $(k+1)$-exchangeable), then
an approximation of the above result holds.

\begin{thm}\label{thm:fqdt} \cite{renner} Assume that $\rho^{(n+k)}$ is a symmetric state
acting on $\mathcal{H}^{\otimes {n+k}}$. Let
$\rho^{(n)}=\Tr_{1\dots k}\,\rho^{(n+k)}$ be the state obtained by
tracing out the first $k$ registers. Then there exists a
probability measure $\mu$ on the set of density matrices on
$\mathcal{H}$ such that
$$\| \rho^{(n)} - \int \mu(d\sigma)\sigma^{\otimes n}\|_{\text{\rm tr}}
\,\leq 2\dim\mathcal{H}\,\frac{n}{n+k}.$$

\end{thm}


\section{Proof of Theorem \ref{thm:main}}\label{sec:proof}

As we mentioned our proof is based on the work of Brandao and Plenio about reversibility of entanglement transformation under asymptotically non-entangling maps. In particular, we follow similar steps as in the proof of Corollary II.2 of \cite{stein}.

Let $\mathcal{H}=\mathcal{H}^A\otimes \mathcal{H}^B$ and assume
that $d = \dim \mathcal{H} > 6$. Then there exists a PPT state
$\rho_{AB}=\rho$ acting on $\mathcal{H}$ which is not separable (see \cite{bound}). Let
\begin{equation} \label{eq:dist} \epsilon=\min_{\sigma\in \SEP} \|\rho
- \sigma\|_{\Tr}.
\end{equation}
Since $\rho$ is entangled, $\epsilon > 0$.

For every number $n$, $\rho^{\otimes n}$ can be considered as a
bipartite state acting on $(\mathcal{H}^A)^{\otimes n }\otimes
(\mathcal{H}^B)^{\otimes n }$, and it is a PPT state. Therefore,
if we prove that the trace distance of $\rho^{\otimes n }$ from
separable states tends to $1$, as $n$ goes to infinity, we are
done.

Let $\sigma^{(n)}$ be the closest separable state to
$\rho^{\otimes n}$. Since $\rho^{\otimes n}$ is a symmetric state,
for any permutation $\pi$ we have
$$\|\rho^{\otimes n} - P_{\pi}\sigma^{(n)} P_{\pi}^{\dagger}\|_{\Tr} = \|\rho^{\otimes n} - \sigma^{(n)}\|_{\Tr},$$
and by triangle inequality
$$\|\rho^{\otimes n} - \frac{1}{n!}\sum_{\pi}P_{\pi}\sigma^{(n)}
P_{\pi}^{\dagger}\|_{\Tr} \leq  \frac{1}{n!} \sum_{\pi} \|\rho^{\otimes n} -
P_{\pi}\sigma^{(n)} P_{\pi}^{\dagger}\|_{\Tr}= \|\rho^{\otimes n} -
\sigma^{(n)}\|_{\Tr}.$$ Therefore, we
may assume that $\sigma^{(n)}$ is symmetric.

Let $\sigma^{(n+n^2)}$ be the closest (symmetric) separable state
to $\rho^{\otimes (n+n^2)}$, and let $\Tr_{1\dots n^2}\,
\sigma^{(n+n^2)}$ be the state obtained by tracing out $n^2$
registers. We have
\begin{equation}\label{eq:10}\| \rho^{\otimes (n+n^2)} -  \sigma^{(n+n^2)} \|_{\Tr}\, \geq
\, \| \rho^{\otimes n} - \Tr_{1\dots n^2}\,
\sigma^{(n+n^2)}\|_{\Tr}. \end{equation} Using the finite quantum
de Finetti theorem (Theorem \ref{thm:fqdt}), there exists a
measure $\mu$ such that
\begin{equation}\label{eq:xn} \Tr_{1\dots n^2}\, \sigma^{(n+n^2)} = \int
\mu(d\tau)\tau^{\otimes n} + X_n,
\end{equation}
where $\|X_n\|_{\Tr}\leq 2d\,\frac{n}{n+n^2}$. Thus using (\ref{eq:10}), if we prove that
\begin{equation}\label{eq:320}
\| \rho^{\otimes n} - \big(\int
\mu(d\tau)\tau^{\otimes n} + X_n\big) \|_{\Tr}
\end{equation} tends to $1$, as
$n$ goes to infinity, we are done.

Consider an informationally complete POVM on $\mathcal{H}^A$ and
$\mathcal{H}^B$, and by taking their pairwise tensor product
extend them to an informationally complete POVM on $\mathcal{H}$.
Then apply quantum state tomography on $(n-1)$ copies of $\rho$ in order to obtain an
approximation of this state.
To be more precise, let $\{M_i\}$ be the resulting informationally
complete POVM on $\mathcal{H}$. So for a sequence of outcomes
$(M_{l_1}, \dots ,M_{l_{(n-1)}})$ we get to the approximation
\begin{equation}\label{eq:tomog-approximation}
\sum_i \frac{r_i}{n-1}M_i^{\ast},
\end{equation}
where $r_i$ is the number of repetitions of $M_i$ in $(M_{l_1},
\dots ,M_{l_{(n-1)}})$.

We say that $(M_{l_1}, \dots ,M_{l_{(n-1)}})$ is a \emph{good} sequence if its corresponding estimation belongs to
$B_{\epsilon/3}(\rho)$, the ball of radios $\epsilon/3$ in trace
distance around $\rho$. Let $G_{n-1}$ be the sum of $M_{l_1}\otimes \cdots \otimes M_{l_{(n-1)}}$ over good
sequences $(M_{l_1},
\dots ,M_{l_{(n-1)}})$. Therefore, by the law of large numbers
\cite{dudley}, $\Tr(G_{n-1} \rho^{\otimes (n-1)}) \rightarrow 1 $ as
$n$ goes to infinity. Also for any $\tau$ far from $\rho$,
$\Tr(G_{n-1} \tau^{\otimes (n-1)})$ tends to zero.

Note that $G_{n-1}\leq I$. Thus
$$\|
\rho^{\otimes n} - \big( \int \mu(d\tau)\tau^{\otimes n} + X_n
\big) \|_{\Tr} \geq \Tr\, ( I\otimes G_{n-1} \cdot \rho^{\otimes n} )
- \Tr \big[ (I\otimes G_{n-1}) \cdot \big(\int \mu(d\tau)\tau^{\otimes
n} + X_n\big) \big].$$
Since $\Tr\, ( I\otimes G_{n-1} \cdot
\rho^{\otimes n} ) \rightarrow 1$, if we prove
$$\Tr \big[ (I\otimes
G_{n-1}) \cdot (\int \mu(d\tau)\tau^{\otimes n} + X_n) \big]
\rightarrow 0,$$ we conclude that (\ref{eq:320}) tends to $1$.

Now suppose that we perform quantum state tomography on  $\int \mu(d\tau)\tau^{\otimes n} + X_n$.
By (\ref{eq:xn}), this state is
not entangled. Moreover, since we can apply tomography
locally (see Section \ref{sec:tomography} ), by starting from a separable state, the
outcome of the process is always separable as well. Assuming that we get a good sequence in the process the outcome is equal to
\begin{equation}\label{eq:420}
\int \mu(d\tau)\,\Tr[G_{n-1}\tau^{\otimes (n-1)}]\,\tau + \widetilde{X}_n,
\end{equation}
where $\|\widetilde{X}_n\|_{\Tr}\leq \|X_n\|_{\Tr} \leq 2d\,\frac{n}{n+n^2}$. As a result, this state is separable. (Here we assume that (\ref{eq:420}) is non-zero because otherwise there is nothing to prove.)

Let
$$Y_n = \int_{\tau\notin B_{\epsilon/2}(\rho)} \mu(d\tau)\Tr[G_{n-1}\tau^{\otimes
(n-1)}]\,\tau +\widetilde{X}_n ,$$ and
$$c_n= \int_{\tau\in B_{\epsilon/2}(\rho)} \mu(d\tau)\Tr[G_{n-1}\tau^{\otimes
(n-1)}].$$ By the law of large numbers, there exists
$\delta_n$ such that for any $\tau\notin B_{\epsilon/2}(\rho)$ we
have
$$\Tr [G_{n-1} \tau^{\otimes (n-1)}]\leq \delta_n,$$
and $\delta_n \rightarrow 0$ as $n$ goes to infinity. Therefore,
$\|Y_n\|_{\Tr}\leq \delta_n + 2d\frac{n}{n+n^2}$.

As we mentioned, the state
$$\widetilde{\tau}= \frac{1}{c_n+\Tr (Y_n)} \left[   \int_{\tau\in
B_{\epsilon/2}(\rho)} \mu(d\tau)\Tr[G_{n-1}\tau^{\otimes (n-1)}]\,\tau
+ Y_n \right]$$ is separable. On the other hand, by definition
$$\widetilde{\rho}= \frac{1}{c_n} \int_{\tau\in
B_{\epsilon/2}(\rho)} \mu(d\tau)\Tr[G_{n-1}\tau^{\otimes (n-1)}]\,\tau
$$ is in the $\epsilon/2$-neighborhood of $\rho$. Then by (\ref{eq:dist}) we have
\begin{eqnarray*}
\epsilon & \leq & \| \rho - \widetilde{\tau} \|_{\Tr}\\
& \leq & \frac{c_n}{c_n+\Tr(Y_n)} \| \rho
-\widetilde{\rho}\|_{\Tr} +
\frac{|\Tr(Y_n)|}{c_n+\Tr(Y_n)}\|\rho\|_{\Tr}+
\frac{1}{c_n+\Tr(Y_n)}\|Y_n\|_{\Tr}\\
& \leq & \frac{c_n}{c_n+\Tr(Y_n)} \cdot \frac{\epsilon}{2} +
\frac{2}{c_n+\Tr(Y_n)}\|Y_n\|_{\Tr}.
\end{eqnarray*}
Thus
$$\epsilon c_n + \epsilon \Tr(Y_n) \,\leq \, \frac{\epsilon}{2}c_n + 2\|Y_n\|_{\Tr},$$
and then
$$ c_n\,\leq \, \frac{2(2+\epsilon)}{\epsilon}\|Y_n\|_{\Tr} \,\leq \, 6 \epsilon^{-1} (\delta_n+ 2d\frac{n}{n+n^2}).$$
Putting everything together we find that
\begin{eqnarray*}
\Tr \big[ (I\otimes G_{n-1}) \cdot (\int \mu(d\tau)\tau^{\otimes n} +
X_n)] & = &\Tr \big[ \int_{\tau\in B_{\epsilon/2}(\rho)}
\mu(d\tau)\Tr[G_{n-1}\tau^{\otimes (n-1)}]\,\tau + Y_n \big] \\
& \leq & c_n + \| Y_n \|_{\Tr} \\
& \leq & (6\epsilon^{-1}+1) \cdot ( \delta_n+ 2d\frac{n}{n+n^2}),
\end{eqnarray*}
which gives
$$\Tr \big[ (I\otimes G_{n-1}) \cdot (\int \mu(d\tau)\tau^{\otimes n} +
X_n)]  \rightarrow 0,$$ as $n$ goes to infinity. We are done.

\section{Geometry of the set of separable states}

Theorem \ref{thm:main} tells us that estimating the distance of a
bipartite state from separable state by the distance from PPT
states is not a good approximation. However, one may expect that the set
of PPT states is an approximation of the set of
separable states from a geometrical point of view. For instance,
two spheres centered at origin with radiuses $1$ and $2$ are far
from each other, while they have the same geometric properties up
to a scaler factor. In the following theorem we bring an evidence that this is not the case
for the set
of separable states relative to PPT states.

By Theorem \ref{thm:main} the maximum distance of a PPT state
from the boundary of the set of separable states is close to $1$.
We can think of this problem in another direction. What is the
maximum distance of a state on the boundary of separable states
from the boundary of PPT states? To get an intuition on this
problem, we can think of the unit sphere centered at origin in
$\mathbb{R}^n$, and the cube with vertices $(\pm 1, \dots , \pm
1)$. It is easy to see that the distance of any point on the
sphere from points of the cube is less than $2$. However, the
distance of $(1, \dots , 1)$ from sphere is $\sqrt{n}-1$. It is
because sphere and cube have totally different shapes.

\begin{thm} \label{thm:shape} Assume that $\mathcal{H}=\mathcal{H}^A\otimes \mathcal{H}^B$, and
$\dim \mathcal{H}^A=\dim \mathcal{H}^B=d$. Then for any separable
state $\rho$ acting on $\mathcal{H}$ there exists a state
$\sigma$ on the boundary of the set of PPT states such that
$\|\rho - \sigma \|_{\Tr}  \leq \frac{1}{\sqrt{d}}.$
\end{thm}

\begin{proof}
Let $\sigma$ be an arbitrary PPT state, and $\Phi(d)$ be the
maximally entangled state defined in (\ref{eq:epr}). Then the
fidelity of $\sigma$ and $\Phi(d)$ is
$$F(\sigma, \Phi(d))=[\Tr\, \sigma\,\Phi(d) ]^{1/2} = [\Tr\, \sigma^{\T_{\!B}}\,\Phi(d)^{\T_{\!B}} ]^{1/2} =
[\Tr\, \sigma^{\T_{\!B}}\, ( \frac{1}{d}I -\frac{2}{d}\sum_{i<j} \vert
\phi_{ij}\rangle \langle \phi_{ij}\vert ) ]^{1/2},$$ where $\vert
\phi_{ij}\rangle $ is defined in (\ref{eq:phi}). Thus using
the fact that $\sigma^{\T_{\!B}}$ is a density matrix and then $\sigma^{\T_{\!B}} \leq I$, we obtain
$$F(\sigma, \Phi(d)) \leq \frac{1}{\sqrt{d}}\,.$$
Therefore, by the well-known inequality between fidelity and
trace distance (see \cite{chuang} page 416) we have
\begin{equation} \label{eq:sqrt} \| \sigma -\Phi(d)\|_{\Tr} \geq 1- F(\sigma, \Phi(d))
\geq 1- \frac{1}{\sqrt{d}} \,.
\end{equation}

Now let $\rho$ be an arbitrary separable state. Define $\rho_t= (1-t)
\rho + t\Phi(d)$. $\rho_0=\rho$ is separable and then PPT,
and $\rho_1=\Phi(d)$. So there exists $0 \leq  c \leq 1$ such
that $\rho_c$ is on the boundary of PPT states. Then we have
$$\|\rho - \rho_c\|_{\Tr} = \|\rho - \Phi(d)\|_{\Tr} -
\|\rho_c - \Phi(d)\|_{\Tr}\leq 1-(1- \frac{1}{\sqrt{d}})=
\frac{1}{\sqrt{d}}\, ,$$ where in the last inequality we use
(\ref{eq:sqrt}).

\end{proof}

This theorem together with Theorem \ref{thm:main} say that considering the sets of separable states and PPT states in the trace-norm space, they have completely different shapes. However, due to Dvoretzky's theorem (see for example \cite{asw}) we know that for every convex set, its intersections with \emph{most} hyperplanes of certain dimension are close to Euclidean ball in shape. This means that the set of separable states and PPT states have the same geometry if we consider them in Euclidean space and restrict them to sections of certain dimension.

\section{Generalization to other separability criteria}

According to Theorem \ref{thm:main}, if the dimension of the
space is large enough, there exists a PPT state far
from separable states. Our candidate for such a
state is $\rho^{\otimes n}$, where $\rho$ is an entangled PPT
state, and in the proof the only property of the set of PPT states that we
use, is that this set is closed under tensor product. Therefore,
the same argument as in the proof of Theorem \ref{thm:main},
gives us the following general theorem.

\begin{thm}\label{thm:general} Assume that $C$ is a necessary \emph{but not sufficient} separability
criterion such that if $\rho$ and $\sigma$ satisfy $C$, then
$\rho\otimes \sigma$ satisfies $C$ as well. Then for any
$\varepsilon>0$ there exists a state $\rho$ that satisfies $C$, and
whose trace distance from separable states is at least
$1-\varepsilon$.
\end{thm}

\begin{proof} Let $\rho$ be an entangled state which satisfies $C$.
Then $\rho^{\otimes n}$ satisfies $C$ as well, and by the proof of
Theorem \ref{thm:main}, the trace distance of $\rho^{\otimes n}$
from separable states tends to $1$ as $n$ goes to infinity.
\end{proof}

\vspace{.3cm}

In the following theorem we prove that all separability criteria
mentioned in Section \ref{sec:prel} satisfy the assumption of
Theorem \ref{thm:general}.

\begin{thm}\label{thm:example} For all separability
criteria mentioned in Section \ref{sec:prel} there exists an
entangled state which passes the test while it is arbitrarily far,
in trace distance, from separable states.
\end{thm}

\begin{proof} By Theorem \ref{thm:general} it is sufficient to
prove that those separability criteria are closed under tensor
product.

\begin{itemize}
\item Reduction criterion: Let $X, Y, Z$ and $W$ be positive
semi-definite matrices such that $X\geq Y$ and $Z\geq W$. Then
$(X-Y)\otimes (Z+W)$ and $(X+Y)\otimes (Z-W)$ are positive
semi-definite. Therefore $X\otimes Z -Y\otimes W = \frac{1}{2}[
(X-Y)\otimes (Z+W) + (X+Y)\otimes (Z-W) ]$ is positive
semi-definite. It means that if $X\geq Y$ and $Z\geq W$, then
$X\otimes Z\geq Y\otimes U$. Now assume that $\rho_{AB}$ and
$\sigma_{A'B'}$ pass the reduction criterion. Therefore $\rho_A\otimes
I \geq \rho_{AB}$ and $\sigma_{A'}\otimes I \geq \sigma_{A'B'}$, and
then $ \rho_{A}\otimes \sigma_{A'} \otimes I \geq \rho_{AB}\otimes
\sigma_{A'B'}$, which means that $\rho_{AB}\otimes \sigma_{A'B'}$ satisfies the
reduction criterion.

\item Entropic criterion: It follows easily from
$S_\alpha(\rho\otimes \sigma)=S_\alpha(\rho)+ S_{\alpha}(\sigma)$.

\item Majorization criterion: $x \prec y$ if and only if there
exists a doubly-stochastic matrix $D$ (a matrix all of whose entries are positive, and the
sum of entries on any row and column is equal to $1$) such
that $x=Dy$ (see \cite{chuang} page 575). Therefore, if $x \prec y$
and $x' \prec y'$, there exist $D$ and $D'$ such that $x=Dy$ and
$x'=D'y'$. Hence $x\otimes x'=(D\otimes D')(y\otimes y')$ and
then $x\otimes x'\prec y\otimes y'$. The proof follows easily
using this property.

\item Cross norm criterion: Using $v(X\otimes X')=v(X)\otimes
v(X')$ we have $\mathcal{U}((X\otimes X')\otimes (Y\otimes Y')) =
\mathcal{U}(X\otimes Y)\otimes \mathcal{U}(X'\otimes Y')$. The
proof follows from this equation.

\item Symmetric extension criterion: If $\rho^{(k)}$ and $\sigma^{(k)}$ are
symmetric extensions of $\rho$ and $\sigma$ to $k$ copies
respectively, then $\rho^{(k)}\otimes \sigma^{(k)}$ is a
symmetric extension of $\rho\otimes \sigma$ to $k$ copies.
\end{itemize}
\end{proof}

\section{Conclusion}

We have proved that for any separability criterion
that is closed under tensor product, the
set of states that pass the test, is not a good approximation of
the set of separable states. For the special case
of positive partial transpose test,
we have shown that the sets of PPT states and separable states have
totally different shapes.
A problem that arises naturally is to
find a separability criterion which is not weaker than the known
ones, and also is not closed under tensor product. Finding such a separability test
may clarify the complexity of the separability problem: is it
$\NP$-hard to decide whether there exists a separable state whose
trace distance from a given state is less than a given constant $c$?

\hspace{.5in}

\noindent {\bf Acknowledgements.} Authors are grateful to Karol {\.Z}yczkowski and Stanis{\l}aw J. Szarek for
providing some background about the comparison of the volume of the sets of separable states and PPT states.  SB is also thankful to Barbara Terhal
for useful discussions. 


\end{document}